\documentclass[conference]{IEEEtran}
\usepackage[numbers, square, comma, sort&compress]{natbib}
\usepackage{amsmath,amsthm,mathtools}
\usepackage{color}
\usepackage{amssymb, bm}
\usepackage{float}
\usepackage{amssymb,amsmath,amsfonts}
\usepackage{rotating}
\usepackage{graphicx}
\usepackage{multirow}
\usepackage{booktabs}
\newcommand{\e}{\epsilon}

\newcommand{\vx}{\mathbf{x}}
\newcommand{\vy}{\mathbf{y}}
\newtheorem{lem}{Lemma}
\newtheorem{thm}{Theorem}

\newtheorem{defin}{Definition}

\theoremstyle{definition}

\newcommand{\suchthat}{\;\ifnum\currentgrouptype=16 \middle\fi|\;}

\begin{document}
%
\title{On the Convergence Speed of Spatially Coupled LDPC Ensembles}

%
\author{\IEEEauthorblockN{Vahid Aref\IEEEauthorrefmark{1}\IEEEauthorrefmark{2},
Laurent Schmalen\IEEEauthorrefmark{2}, and
Stephan ten Brink\IEEEauthorrefmark{2}, 
\IEEEauthorblockA{\IEEEauthorrefmark{1}School of Information and Communication, EPFL
Lausanne, Switzerland\\ Email: vahid.aref@epfl.ch}
\IEEEauthorblockA{\IEEEauthorrefmark{2}Bell Laboratories, Alcatel-Lucent, Stuttgart, Germany\\
Email: \{Laurent.Schmalen,Stephan.tenBrink\}@alcatel-lucent.com}}}
\maketitle
\begin{abstract}
Spatially coupled low-density parity-check codes show an outstanding performance
under the low-complexity belief propagation (BP) decoding algorithm.
They exhibit a peculiar convergence phenomenon above
the BP threshold of the underlying non-coupled ensemble,
with a wave-like convergence propagating through the spatial dimension
of the graph, allowing to approach the MAP threshold.
We focus on this particularly interesting
regime in between the BP and MAP thresholds.

On the binary erasure channel, it has been proved \cite{kudekar2012wavesolution}
that the information propagates with a constant speed toward the successful decoding solution.
We derive an upper bound on the propagation speed,
only depending on the basic parameters of the 
spatially coupled code ensemble
such as degree distribution and the coupling factor $w$.
%
We illustrate the convergence speed of different code ensembles by simulation results,
and show how optimizing degree profiles helps to speed up the convergence.  
\end{abstract}

\IEEEpeerreviewmaketitle
\section{Introduction}
Low-density parity-check (LDPC) codes are widely used due to their
outstanding performance under low-complexity belief propagation (BP)
decoding.  However, an error probability exceeding that of
maximum-a-posteriori (MAP) decoding has to be tolerated with
(sub-optimal) BP decoding. However, it has been empirically observed
for spatially coupled LDPC (SCLDPC) codes -- first introduced by
Felstr\"om and Zigangirov as convolutional LDPC codes \cite{ZigFel} --
that the BP performance of
these codes can improve dramatically towards the MAP performance of the
underlying code under many different settings and conditions, e.g.
\citep{TerminLDPCCCthreshold,IDLDPCC,ProtoLDPCC}.
 This phenomenon --
termed \emph{threshold saturation} -- has been proven by Kudekar,
Richardson and Urbanke in \citep{CouplLDPC09,Coupl11BMS}. In
particular, they proved that the BP threshold of a coupled LDPC
ensemble tends to its MAP threshold on any binary symmetric memoryless
channel (BMS).  The principle behind threshold saturation seems to be
very general and has been applied to a variety of more general
scenarios in information theory and computer sciences.

More recently, a new proof technique of threshold saturation has been given
by introducing a potential function~\citep{yedla2012scaler,kumar2013bms}.
Independently, a similar technique has been used in
\cite{kudekar2012wavesolution} to study one-dimensional continuous coupled
systems. It was proven that below the MAP threshold,
the information propagates in a wave-like manner with a constant propagation
speed 
by progressing density evolution.
This result has also been extended to discrete spatially coupled
one-dimensional systems.

We can distinguish between two convergence regions for SCLDPC
codes. If the channel entropy is below the BP threshold of the
underlying non-coupled ensemble, the convergence is governed by the
ensemble degree distribution. We call this convergence \emph{intra-graph convergence}. If
the channel entropy is between the BP and MAP thresholds of the
non-coupled ensemble, the wave-like solutions manifests itself after some iterations and the
convergence -- which we denote \emph{inter-graph convergence} -- is
dominated by the coupling of the graphs. We particularly focus on this latter convergence phenomenon and are interested in the \emph{inter-graph convergence (propagation) speed} of the wave-like solution.
Knowing this speed has several important practical implications: The
degree distribution can be optimized in order to maximize the speed for a given
target channel and consequently, to minimize the number of required
iterations to decode successfully. Additionally, if we
employ windowed decoding~\citep{iyengartitwindow}, the convergence speed
gives the correct timing to shift the decoding window.

In this paper, we derive upper bounds on the convergence speed of spatially coupled LDPC ensembles for the BEC. Moreover, we compare the convergence speeds of
different ensembles and show that the speed is sensitive to the choice of
the degree distribution. We also compare the speed of a spatially coupled
ensemble on different types of BMS channels.

 
\section{Preliminaries}\label{sec:background}
We briefly explain the graphical structure of LDPC ensembles and spatially coupled LDPC ensemble. 
Then, we describe the potential functions associated to these ensembles.
\subsection{LDPC Ensembles}
LDPC codes are a subset of block codes with sparse parity check matrices.
Let $n$ be the block length of the code and $m$ be the number of constraints to satisfy. The design rate of the code is $r=1-m/n$.
We usually represent LDPC codes by a bipartite graph called factor graph. 
 To each of the $n$ bits, we assign a node, called variable node 
 and we assign a node to each of the $m$ constraints, called check node.
 We connect variable node $i$ to check node $a$ by an edge if and only if 
 the bit $i$ participates in the corresponding constraint.
To construct a random LDPC code, we sample the degree of node according to a given degree distribution. We represent the degree distribution of the variable 
nodes by a polynomial $L(x)=\sum_{d=1}^{\infty} L_d x^d$ and the degree distribution of check nodes by $R(x)=\sum_{d=1}^{\infty} R_d x^d$. For a node of degree $d$, we consider $d$ sockets from which the $d$ edges emanate. Thus,
There are $L'(1)=m R'(1)$ sockets in both variable side and check side. After labeling the sockets, we randomly choose a socket from the variable side and connect to a randomly chosen socket in the check side. Finally, we have 
a random instance of a LDPC$(L,R)$ ensemble. For the detail of construction, we refer to \cite{URbMCT}.

To study the performance of belief propagation algorithm on sparse graph codes, it is common to use \emph{density evolution}. Let
$\e$ denote the erasure probability of the channel and let $x^{(t)}$ denote the erasure probability flowing from variable side to the check side at iteration $t$, then
\begin{equation}
\label{eq:de_single}
x^{(t+1)} = \e \lambda(1-\rho(1-x^{(t)})),
\end{equation}
where 
$\lambda(x)=L'(x)/L'(1)$ and $\rho(x)=R'(x)/R'(1)$. We set $x^{(0)}=1$.
The error probability of decoding vanishes if 
\begin{equation*}
x^{(\infty)}=\lim_{t\to\infty} x^{(t)} = 0.
\end{equation*}
\begin{defin}[BP Threshold]
The BP threshold of the LDPC$(L,R)$ ensemble is defined as:
\begin{equation*}
\e_{\rm BP} = \sup\{\e\in[0,1] \mid x^{(\infty)}(\e)=0\}.
\end{equation*}
\label{def:bp_threshold}
\end{defin}
The MAP threshold $\e_{\rm MAP}$ is defined as the maximum $\e$ in which the decoding error probability of MAP decoding is equal to zero. In general, $\e_{\rm BP}\leq \e_{\rm MAP}$.

\subsection{Spatially Coupled LDPC Ensemble}
We first lay out a set of positions indexed by integers $z\in \mathbb{Z}$
on a line. This line represents a
 \emph{spatial dimension}. We fix a \emph{coupling factor} 
 which is an integer $w> 1$. 
 Consider $2N$ sets of
  variable nodes each having $n$ nodes, and locate the 
  sets in positions $1$ to $2N$. Similarly, locate $2N + 
  w-1$ sets of $m$ check nodes each, in positions $1$ to 
  $2N+w-1$. The degree of each variable (check) node is 
  randomly sampled according to the degree distribution $L(x)$ 
  ($R(x)$) leading to $m R'(1)$ sockets at 
  each position. 
  
  To construct a random instance of the 
  SCLDPC$(L,R,N,w)$ ensemble, 
  we connect the variable nodes to the 
  check nodes in the following manner: Each of sockets 
  of variable nodes at position $z\in\{1,\dots,2N\}$ is 
  connected uniformly at random to a socket of check 
  nodes within the range $\{z,\dots,z+w-1\}$. At the end 
  of this procedure, all $mR'(1)$ sockets of check nodes in position $z\in [w,2N]$ are occupied except for the check nodes in $z\in[1,w-1]$, where only a fraction $z/w$ of them are connected. Symmetrically, for the check nodes in $z\in[2N+1,2N+w-1]$, a fraction $(w+2N-z)/w$ of the sockets are connected. The rest of the sockets are free
and we can assume that they are connected to virtual variable nodes with zero erasure probability. For the detail of construction, we refer to \cite{CouplLDPC09}.

Let $x_z^{(t)}$ be the erasure probability incoming to check nodes in position $z\in\mathbb{Z}$ at iteration $t$. The density evolution (DE) equation is
\begin{equation}
x_z^{(t+1)} = \frac{1}{w} \sum_{k=0}^{w-1} \e_{z-k} \lambda\left(1- \frac{1}{w}
\sum_{j=0}^{w-1}\rho(1-x_{z+j-k}^{(t)})\right)
\label{eq:decoupl},
\end{equation}
where $\e_z=\e$ for $z\in[1,2N]$ , and zero otherwise. 
We initialize $x_z^{(0)}=1$ for $z\in [1,2N+w-1]$. For the boundary 
values, $z\notin[1,2N+w-1]$, we set $x_z^{(t)}=0$ for all $t$.

As density evolution progresses, the perfect boundary information
from the left and right sides propagates inward. 
It was shown in
\citep{CouplLDPC09,Coupl11BMS} that $x_z^{(t)}$ is non-decreasing 
sequence for $z\le N+\lfloor \frac{w-1}{2}\rfloor$ and becomes 
non-increasing sequence afterward. 
The symmetric initialization and symmetric boundary conditions
induce symmetry on all the erasure probabilities, i.e. 
$x_z^{(t)}= x_{2N+w-z}^{(t)}$. This system has been termed \emph{two-sided
spatially coupled LDPC ensemble}. One half of the spatially coupled ensemble is enough to describe the system. 
\begin{defin}[\citep{CouplLDPC09}]
Let $N'=N+\lfloor \frac{w-1}{2}\rfloor$. The one-sided spatially coupled LDPC ensemble is a modification of \eqref{eq:decoupl} defined by fixing
$x_z^{(t)} = x_{N'}^{(t)}$ for $z\ge N'$.
\end{defin}
\begin{lem}[\cite{CouplLDPC09}]
For the one-sided spatially coupled LDPC ensemble, 
the densities resulting from density evolution over the BEC satisfy
\begin{equation*}
x_z^{(t)}\ge x_{z-1}^{(t)}, 
\end{equation*}
and
\begin{equation*}
x_z^{(t)}\ge x_{z}^{(t+1)}, 
\end{equation*}
for all $z\in\mathbb{Z}$ and $t\ge 0$,\label{lem:monotonicity}
\end{lem}

\subsection{Potential Function}
We define the potential function of the LDPC$(L,R)$ ensemble as in \cite{yedla2012scaler},
\begin{multline}
\label{eq:potsingle}
U(x;\e)= \frac{1}{R'(1)} \left(1-R(1-x)\right) - x\rho(1-x) 
\\- \frac{\e}{L'(1)} L(1-\rho(1-x))
\end{multline}
It is indeed equal to the normalized (by $-L'(1)$) replica-symmetric free energies in \cite{macris2007griffith} and it has the following properties.
\begin{lem}\label{lem:pot_properties}
Consider the potential function in \eqref{eq:potsingle}:
\begin{itemize}
\item[i)] $\frac{\partial}{\partial x} U(x;\e) = \rho'(1-x)\left(x-\e \lambda\left(1-\rho(1-x)\right)\right)$.
\item[ii)] For $\e<\e_{\rm BP}$ and $0<x\le 1$, $\frac{\partial}{\partial x} U(x;\e)>0$.
\item[iii)] $U(x;\e)$ is strictly decreasing in terms of $\e$.
\item[iv)] The stationary points of the potential function are the
fixed-points of DE equation \eqref{eq:de_single}. 
\end{itemize}
\end{lem}
Depending on the degree distributions $L(x)$ and $R(x)$, the potential 
function can have many stationary
points. The potential function of regular LDPC$(x^3,x^6)$ ensemble for $\e=0.475$ is depicted in
Fig.~\ref{fig:pot}. Each local minimum corresponds to a stable fixed-point
of DE equation and each local maximum corresponds 
to an unstable fixed-point. For $\e<\e_{\rm BP}$ there is only one
local minimum which is $x=0$.

\begin{figure}[tb]
\setlength{\unitlength}{0.72
bp}\centering
\input{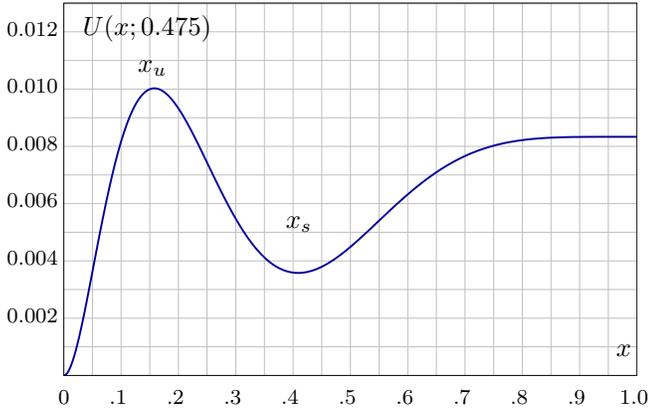}
\caption{ The potential function of regular LDPC$(x^3,x^6)$ ensemble for $\e=0.475$. The stable and unstable fixed-points of DE equation are illustrated in the plot.}
\label{fig:pot}
\end{figure}

\begin{defin}
We define the area threshold as
\begin{equation*}
\e_{\rm area} = \sup\{ \e\in[0,1]\mid U(x;\e)\ge 0,
\;{\rm for }\; x\in[0,1]\}.
\end{equation*} 
\end{defin}
It is shown in \cite[Lemma 6]{yedla2012scaler} that the area threshold is 
equal to the Maxwell threshold. The Maxwell threshold is equal to MAP threshold on regular ensembles. It has been conjectured that they are equal in general and 
the equality was recently justified for a large class of LDPC codes in \cite{andrieProof}.
Thus, $U(x;\e)>0$ for $\e<\e_{\rm MAP}$ and becomes zero at $\e_{\rm MAP}$.

Now consider the SCLDPC$(L,R,N,w)$ ensemble.  
We use the extension of potential function introduced in \cite{yedla2012scaler} as follows:
\begin{align}
U(\vx;e) = \sum_{z=1}^{N'} 
&\frac{1}{R'(1)} \left(1-R(1-x_z)\right) - x_z\rho(1-x_z)\nonumber\\
&-\frac{\e}{L'(1)} L\left(1-\frac{1}{w}\sum_{j=0}^{w-1}
\rho(1-x_{z+j})\right),
\label{eq:potcoupl}
\end{align}
where $\vx=(x_1,\dots,x_{N'})$. 
We retrieve the DE equation from the partial 
derivatives of the potential function, i.e.
\begin{align*}
\frac{\partial}{\partial x_z} &U(\vx;e)
= \rho'(1-x_z) \times\\
&\left(x_z - \frac{1}{w} \sum_{k=0}^{w-1} \e_{z-k} \lambda\left( 1-\frac{1}{w}
\sum_{j=0}^{w-1} \rho(1-x_{z+j-k})\right)\right).
\end{align*}
Thus, the stationary points of $U(\vx;e)$ are 
the fixed-points of DE equation.

Define the vector $\mathbf{y}$ as $y_z = x_z^{(t)} + h (x_z^{(t+1)}-x_z^{(t)})$ for $0 \le h\le 1$. By using the first order Taylor expansion,
\begin{equation*}
U(\mathbf{y};\e)-  U(\mathbf{x}^{(t)};\e)= \sum_{z=1}^{N'}\frac{\partial U(\vx^{(t)};e)}{\partial x_z}  (y_z - x_z^{(t)}) + R_1(\mathbf{y},\vx^{(t)}),
\end{equation*}
where $R_1(\mathbf{y},\vx^{(t)})$ is the remainder term. 
Define
\begin{equation*}
\Delta U_1(\mathbf{y};\mathbf{x}^{(t)})
= \sum_{z=1}^{N'}\frac{\partial U(\vx^{(t)};e)}{\partial x_z}  (y_z - x_z^{(t)}).
\end{equation*}
Numerical evaluation shows that the remainder term is negative and small in comparison with $\Delta U_1(\mathbf{y};\mathbf{x}^{(t)})$ and then,
\begin{equation*}
U(\mathbf{y};\e)-  U(\mathbf{x}^{(t)};\e)\approx \Delta U_1(\mathbf{y};\mathbf{x}^{(t)}).
\end{equation*}
To upper-bound the speed, we must show that there is $\alpha\geq 1$ such that
for all $t$ and $\mathbf{y}$.
\begin{equation}
\alpha\left(U(\mathbf{y};\e)-  U(\mathbf{x}^{(t)};\e)\right)\leq
\Delta U_1(\mathbf{y};\mathbf{x}^{(t)}).
\label{eq:firstorderapprox}
\end{equation}
We prove $\alpha\leq 2$ in the next section and in Appendix~\ref{apn:taylor} for large $w$. However,
our simulation results suggest $\alpha=1$. The extension to $\alpha=1$ is currently work in progress.


\section{Bounds on The Convergence Speed}\label{sec:bound}

Consider the DE equation of one-sided SCLDPC$(L,R,N',w)$ ensemble. Denote
the BP threshold and the MAP threshold of the underlying LDPC$(L,R)$ ensemble by
$\e_{\rm BP}$ and $\e_{\rm MAP}$, respectively.
It has been proven in \citep{CouplLDPC09,yedla2012scaler} that 
there is $w_0$ such that for $\e<\e_{\rm MAP}$ and $w>w_0$,
\begin{equation*}
\lim_{t\to\infty}\lim_{N'\to\infty}\lim_{n\to\infty} x_z^{(t)}=0,
\end{equation*}
for all $z\in\mathbb{R}$. The question is how fast $x_z^{(t)}$
converges to zero. Assume that $\e_{\rm BP}<\e<\e_{\rm MAP}$.
We distinguish two distinct phases during DE progress.
In the first few iterations, all $x_z^{(t)}$ except the ones close to
$z=0$ converges to the forward DE fixed-point of the underlying ensemble, $x_{\rm BP}\ne 0$ (intra-graph convergence). Then, in the next iterations, the information propagates from the boundary $z=0$ and $x_z^{(t)}$ becomes
zero successively (inter-graph convergence).  

In this section, we bound the speed of information propagation.
First we consider LDPC ensembles whose DE equation has three fixed-points (two stable and one unstable fixed-point). The potential function of one such ensemble is shown in Fig.~\ref{fig:pot}.  
Many LDPC ensembles including regular LDPC codes have such property. Then
we consider the ensembles with more than three DE fixed-points.

\subsection{DE Equation with Three Fixed-points}

\begin{figure}[tb]
\setlength{\unitlength}{0.72
bp}\centering
\input{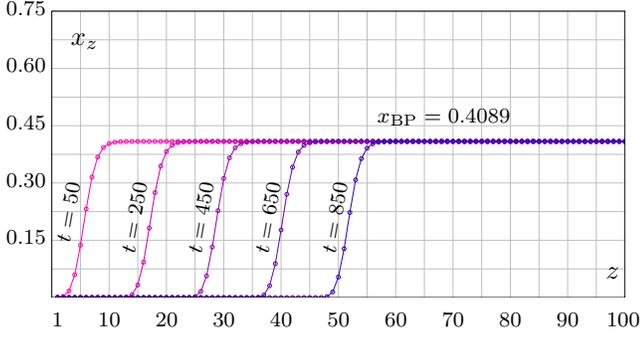}
\caption{ The solution of DE equation for one-sided SCLDPC$(x^3,x^6,100,3)$ ensemble on BEC with $\e=0.475$ in different iterations $t=50, 250, 450, 650$ and $850$. $x_{\rm BP}= 0.4089$.
By progressing density evolution, the wave-like solution moves uniformly forward with a fixed speed.}
\label{fig:wave} 
\end{figure}
Fig.~\ref{fig:wave} shows the solution of DE equation of one-sided SCLDPC$(x^3,x^6,100,3)$ ensemble in different iterations. We observe that the 
sequence $x_z^{(t)},$ $z\in\mathbb{Z}$ moves uniformly forward with a fixed speed by progressing density evolution. More precisely,
there is $T\in\mathbb{N}$ such that 
\begin{equation*}
x_z^{(t+T)}\leq x_{z-1}^{(t)}< x_z^{(t+T-1)},
\end{equation*}
for all  $t>0$ and $z\in\mathbb{Z}$. 
The existence of such a DE solution is proven in
\cite{kudekar2012wavesolution}
for the one-sided SCLDPC$(L,R,N'=\infty,w)$ ensemble 
on $\e_{\rm BP}<\e<\e_{\rm MAP}$ in which the DE equation of the underlying ensemble 
has three fixed-points: 
zero, an unstable fixed-point $x_u$ and a stable fixed-point $x_s = x_{\rm BP}$, the forward DE fixed-point.

We define the propagation speed of such a DE solution:

\begin{defin} For $I\in\mathbb{N}$, define
\begin{equation*}
T_I = \min\{T\in\mathbb{N}\mid x_z^{(t+T)}\leq x_{z-I}^{(t)}, \;{\rm for}\;t>0\;{\rm and}\;z\in\mathbb{Z} \}.
\end{equation*}
Additionally, we define the propagation speed $v_I=I/T_I$.
\end{defin}
One can show that $v_I$ is an increasing sequence and for 
\begin{equation*}
\frac{1}{T_1}=v_1\leq v_I< \frac{1}{T_1-1}.
\end{equation*}
 
We are mostly interested in knowing the speed 
for $\e$ close to $\e_{\rm MAP}$ in which 
the DE solution travels  very slowly, i.e. $T_1\gg 1$.
In the following theorems, we upper bound $v_1$ (or equivalently,
lower-bound $T_1$). We assume that at $t=0$, the wave-like solution is already formed and   
for $N'\to\infty$, 
\begin{equation*}
\lim_{z\to N'} x_z^{(t)}=x_{\rm BP},
\end{equation*}
for some iteration $t$.
\begin{lem}\label{lem:taylor}
Consider \eqref{eq:firstorderapprox}. Assume that $T_1>1$, 
then there exists $w_0$ such that for $w>w_0$, $\alpha\leq 2$.
\end{lem}
The sketch of the proof is given in Appendix \ref{apn:taylor}.

\begin{thm}[Upper bound]
Assume that the DE solution $x_z^{(t)}$ moves forward uniformly and
$x_w^{(t)}=0$ for $t>t_0$. Then, for any $t>t_0$,
\begin{equation}
v_1\leq B_1:=\frac{\alpha U(x_{\rm BP};\e)}{\sum_{z\in\mathbb{Z}} \rho'(1-x_z^{(t)})(x_z^{(t)}-x_{z-1}^{(t)})^2},
\end{equation}
\label{thm:upperbound1}
\end{thm}\emph{Proof:} Appendix \ref{apn:upperbound1}.

As we see later, the above bound is a tight upper bound for large $w$.
For the calculation, we must run DE until the wave-like solution is formed, which in practice, appears after a few iterations. 
We can additionally state the following (looser) upper bound which eliminates the need to run density evolution:
\begin{thm}
Assume that $T_1> 1$. Let $D=\max_{x\in(0,x_{\rm BP})} \vert\frac{\partial^2}{\partial x^2} U(x;\e)\vert$. Then,
\begin{equation}
B_1\leq \frac{w \alpha U(x_{\rm BP};\e)}{2U(x_u;\e)-U(x_{\rm BP};\e) - D\frac{x^2_{\rm BP}}{w}}.
\end{equation}
which simplifies for $w\to\infty$ to
\begin{equation*}
v_1\leq B_2:=\frac{w \alpha U(x_{\rm BP};\e)}{2U(x_u;\e)-U(x_{\rm BP};\e)}.
\end{equation*}
\label{thm:upperbound2}
\end{thm}\emph{Proof:} Appendix \ref{apn:upperbound2}

\begin{thm}[Lower Bound~{\cite[Theorem 2]{kudekar2012wavesolution}} ]
The speed $v_1$ is lower-bounded by
\begin{equation}
v_1\ge \frac{w U(x_{\rm BP})}{x_{\rm BP} (1- \rho(x_{\rm BP}))}-\frac{1}{w}.
\end{equation}
For $w\gg 1$, $v_1\geq LB:= w U(x_{\rm BP})/ \left(x_{\rm BP} (1- \rho(x_{\rm BP})\right)$.
\label{thm:lowerbound}
\end{thm}
Note that $U(x_{\rm BP})/(x_{\rm BP} (1- \rho(x_{\rm BP})))$ is equal to the area defined in \cite{kudekar2012wavesolution}.

\subsection{DE Equation with Many DE Fixed-points}

In general, the DE equation of irregular LDPC 
ensembles can have more than three fixed-points in some $\e$. In this case,
the DE solution of coupled ensemble can become more complex.
It can be a mixture of wave-like solutions with distinct speeds, in which the solution with larger speed overlaps the solutions with lower speed.
For instance, consider
one-sided SCLDPC$(L,R,N'=600,w=3)$ ensemble with $L(x)=\frac{153}{283}x^2 +
\frac{102}{283}x^3 + \frac{28}{283}x^{51}$ and $R(x)=x^{16}$. For $\e_{\rm BP}\approx 0.353<\e<\e^*\approx0.399$,
the underlying ensemble has three DE fixed-points and thus, the DE solution is as explained in the previous section. For $\e^*<\e<\e_{\rm MAP}\approx0.403$, the underlying ensemble has five DE fixed-points, namely 0, $s_1$ and $s_2=x_{\rm BP}$ the stable fixed-points and $u_1$ and $u_2$ the unstable ones. In this case, two wave-like DE solutions appear. In Fig.~\ref{fig:2wave}, the DE solution for $\e=0.4$ is depicted. We observe that  for each $z$, $x_z^{(t)}$
 first reduces from $s_2$ to $s_1$ by the first 
wave-like solution and then, reduces from $s_1$ to $0$ by the second wave-like solution.
\begin{figure}[tb]
\setlength{\unitlength}{0.72
bp}\centering
\input{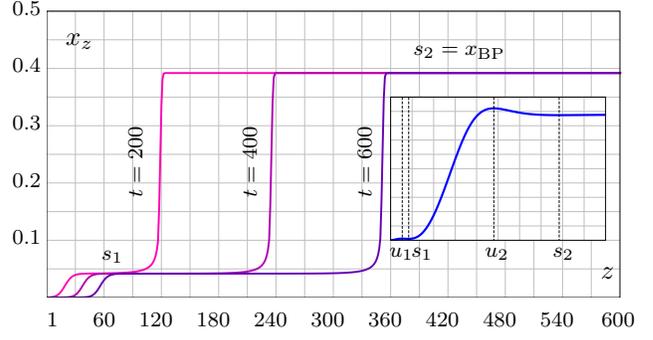}
\caption{ The DE solution of one-sided SCLDPC$(L,R,600,3)$ ensemble with $L(x)=\frac{153}{283}x^2 +
\frac{102}{283}x^3 + \frac{28}{283}x^{51}$ and $R(x)=x^{16}$ on BEC $\e=0.4$ in different iterations $t=200, 400, 600$.
The inlet shows the potential function of LDPC$(L,R)$ for $\e=0.4$.}
\label{fig:2wave} 
\end{figure}

Let $v^u_1$ and $v^l_1$ denote the propagation speed of the upper and 
the lower solutions ($v^u_1\geq v^l_1$).
  $v^l_1$ determines the total number of iterations required for decoding.
We can still apply $B_1$ to upper bound $v^l_1$ with the following assumption.
The waves are well-separated such that the right tail of the lower wave converges to $s_1$ from below and the left tail of upper wave converges to 
$s_1$ from above (see Fig.~\ref{fig:2wave}). If this assumption holds at some $t$, we can separately study the waves by following the proof of Theorem \ref{thm:upperbound1}. Shortly,
\begin{align*}
\alpha (U(s_2;\e)\!-\!U(s_1;\e))&\geq v^u_1\sum_{\mathclap{x_z\in(s_1,s_2]}}
\rho'(1-x_z^{(t)})(x_z^{(t)}-x_{z-1}^{(t)})^2\\
\alpha U(s_1;\e)&\geq v^l_1 \sum_{\mathclap{x_z\in(0,s_1]}}
\rho'(1-x_z^{(t)})(x_z^{(t)}-x_{z-1}^{(t)})^2
\end{align*}
and by summing up both inequalities, and noting that $v_1^u\geq v_1^l$, we have
\begin{align*}
\alpha U(x_{\rm BP};\e)=\alpha U(s_2;\e)\geq v^l_1 \sum_{z\in\mathbb{Z}}
\rho'(1-x_z^{(t)})(x_z^{(t)}-x_{z-1}^{(t)})^2.
\end{align*}
One can also adapt $B_2$ for the general case by following 
the similar proof.

\section{Simulation Results}\label{sec:simulation}
\begin{figure}[tb]
\setlength{\unitlength}{0.72bp}\centering
\input{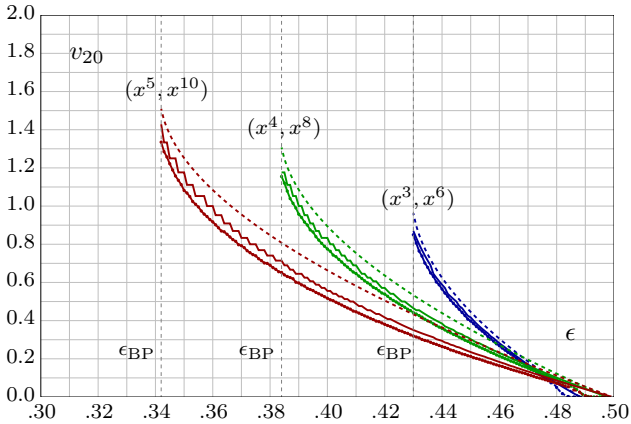}
\caption{ $v_{20}$ vs. $\e$ for some regular SCLDPC$(x^k,x^{2k},N,w=6)$ ensembles and $\e_{\rm BP}<\e<\e_{\rm MAP}$ in each case. $v_{20}$ is computed by
DE (solid curves), and also by running BP over code instances with $N=100$ and $n=5000$. The
results are averaged over 100 instances (dotted curves). The upper bound 
$B_1$ is also shown by dashed curves.}
\label{fig:regularspeed}
\end{figure}
As we mentioned before, 
the numerical results support the conjectured claim $\alpha=1$. In fact,
it is the speed given by the first 
order Taylor approximation. We therefore set $\alpha=1$ in the sequel.  
Consider the SCLDPC$(x^3,x^6,N,w)$ ensemble. We give $v_1$ and $v_{20}$ for $\e=0.475$ and different $w$
in Tab.~\ref{table:36}.
We observe that the speed increases
almost linearly in terms of $w$. 
Note that $v_I$ is an increasing sequence
 in terms of $I$ and saturates to the real propagation speed.
 However, for $v_1\ll 1$, 
$v_1\approx v_I$ for any $I\in\mathbb{N}$. In table
 \ref{table:36}, 
We also compute 
the upper bounds and the lower bound for each $w$. 
We observe that $B_1$ becomes tighter (to $v_I$ with large $I$) 
by increasing $w$. $B_2$ seems to be larger than the real speed by
an almost constant 
factor which does not vary significantly with $\e$.

\begin{table}[tbh]
\centering
\caption{Propagation speed of SCLDPC$(x^3,x^6,N,w)$ ensemble
with different $w$ sizes on BEC with $\e=0.475$.}
\label{table:36}
\renewcommand{\arraystretch}{1.3} \scalebox{1}{\begin{tabular}{@{}cccccc@{}}\toprule

$w$ & 2 & 4 & 8 & 16 & 32\\
\midrule
$v_1$ & 0.035 & 0.077 & 0.143 & 0.250 & 0.333\\
$v_{20}$  & 0.035 & 0.079 & 0.160 & 0.318 & 0.625\\
\midrule
$B_1$ & 0.0503 & 0.091 & 0.172 & 0.333 & 0.656\\
$B_2$  & 0.435 & 0.869 & 1.738 & 3.477 & 6.955\\
\midrule
$LB$ & 0.019 & 0.038 & 0.075 & 0.151 & 0.302\\

\bottomrule
\end{tabular}

}
\end{table}

Now consider the asymptotically regular SCLDPC$(x^k,x^{2k},N=100,w=6,n=5000)$ ensembles. We plot $v_{20}$ and $B_1$ for 
$\e_{\rm BP}<\e<\e_{\rm MAP}$ and 
for $k=3,4$ and $5$ in Fig.~\ref{fig:regularspeed}.  The curves of $B_1$ are shown 
by the dashed lines. We compute $v_{20}$ in two distinct ways: first by the DE equation which are shown by solid lines.
We observe that it is a decreasing function and becomes zero at $\e_{MAP}$. 
Since $T_{20}$ is an integer value and hence, it is not so sensitive
to small changes in $\e$, $v_{20}$ is a staircase function.
Second, we compute the average $v_{20}$ by running BP over
100 code instances. 
These curves
are shown by dotted lines.
As we expect, the empirical $v_{20}$ is close to $v_{20}$ derived by
DE. However,
 the curves deviate from each other for $\e$ close to the MAP threshold.
The reason is due to the error floor that appears for such $\e$ as a result of finite length effects (finite $n$).
The dashed-dotted tail of each
curve shows where the error-floor appears.

Figure~\ref{fig:regularspeed} illustrates that the order of ensembles in terms of speed changes by $\e$. It suggests
 that for a desired $\e$, we can change the degree 
 distributions in order to obtain a faster convergence speed. To be more fair,
 we compare two LDPC ensembles with the same average degree. 
 In Fig.~\ref{fig:degree4curves}, we depict $v_{20}$ for the SCLDPC$(\frac{19}{20}x^3 + \frac{1}{20} x^{23},x^8,N,w=3)$ ensemble
 and for the regular SCLDPC$(x^4,x^8,N,w=3)$ ensemble. We observe that the former has much larger speed than the latter over a wide range of $\e$. For instance, the speed
 is about $1.8$ times larger at $\e=0.475$. However, the regular ensemble has a larger speed for $\e$ very close to its $\e_{\rm MAP}$. 
 \begin{figure}[tb]
\setlength{\unitlength}{0.72bp}\centering
\input{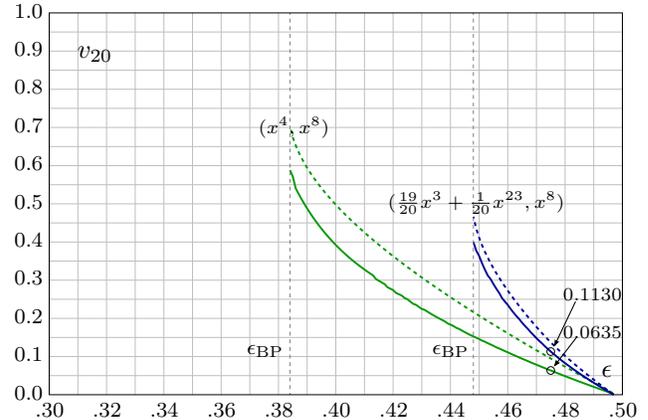}
\caption{ $v_{20}$ vs. $\e$ for regular SCLDPC$(x^4,x^{8},N,w=3)$ ensemble and for SCLDPC$(L(x),x^8,N,w=3)$ ensemble with $L(x)=\frac{19}{20}x^3 + \frac{1}{20}x^{23}$
  and $\e_{\rm BP}<\e<\e_{\rm MAP}$  in each case. Both ensembles have the same average variable degree equal to $4$. $v_{20}$ is computed by
DE (solid curves). The upper bound 
$B_1$ is also shown by dashed curves.}
\label{fig:degree4curves}
\end{figure}

These examples show the need for optimizing the speed for a given $\e$.
 For this purpose, one can use $B_1$ 
 since it is tight for $\e$ close to the MAP threshold.
  
\subsection{Other Types of BMS Channels}
Consider the SCLDPC$(x^3,x^6,N=100,3,n=5000)$ ensemble.  Similar to
the BEC, simulations results suggest that belief propagation also leads to a
wave-like solution on the BSC and on the AWGN channel. In Fig.~\ref{fig:bms}, we
compare the empirical speeds for the BEC, BSC and AWGN channel.  We compute $v_{20}$
for the channel entropy between BP threshold and MAP threshold in each
case.  Each point is averaged over 100 code instances.  We refer to
\cite[Chapter 4]{URbMCT} for the definition of channel entropy.  As
discussed before, we observe an error floor close to the MAP threshold
due to finite length effects.  The dashed tail of each curve
indicates where the error-floor appears.
\begin{figure}[tb]
\setlength{\unitlength}{0.72bp}\centering
\input{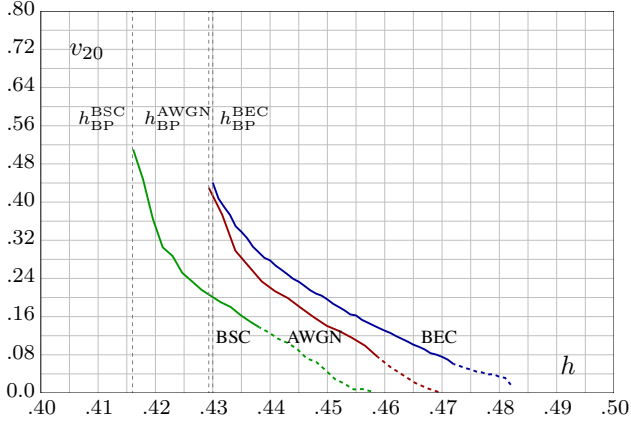}
\caption{ $v_{20}$ vs. channel entropy $h$ for regular SCLDPC$(x^3,x^{6},N,w=3)$ ensembles over BEC, BSC and AWGN channel.
For each channel, $h$ varies between BP threshold and MAP threshold of underlying 
ensembles. $v_{20}$ is computed by running BP over code instances with $N=100$ and $n=5000$. The
results are averaged over 100 instances. The dashed tail of each curve shows where error floor appears. The speed on the BEC is an upper bound over the speed on other channels.}
\label{fig:bms}
\end{figure}

The BP threshold of the underlying ensemble is larger for the BEC than
for the AWGN channel and it is furthermore larger for the AWGN channel
than for the BSC.  The MAP thresholds and the speed preserve the same
order. Thus, we conjecture that the upper bound of the speed on BEC is
also an upper bound of the speed on the other channels with the same
channel entropy.

\section{Conclusion}
In this paper we have analyzed the convergence speed of the wave-like solution that dominates the convergence of spatially coupled codes if the channel entropy lies between the BP and MAP thresholds. Using the Taylor expansion of the spatially coupled code's potential function, we have derived a new upper bound on the convergence speed. This upper bound is based on the degree distribution of the code ensemble and requires the wave-like DE solution. We have additionally derived a looser upper bound that solely depends on the potential function, i.e., on the degree distribution of the ensemble and on the coupling factor. Finally, we have compared the bound with the convergence speeds obtained by density evolution and by simulation of code instances. We have observed that the upper bound seems to be an upper bound  for any binary input memoryless symmetric channel. 

\begin{appendices}
\section{Sketch of Proof of Lemma \ref{lem:taylor}}\label{apn:taylor}
Since $T_1>1$,
\begin{multline}
x_z^{(t)}-x_z^{(t+1)}< x_z^{(t)}-x_{z-1}^{(t)}=
\frac{1}{w}\!\!\left(\!\!\e_{z}\lambda(\frac{1}{w}\sum_{j=0}^{w-1}
\!\!1\! -\!\rho(1\!-\!x_{z+j}^{(t-1)}))\right. \\
\left.-\e_{z-w}\lambda(\frac{1}{w}\sum_{j=0}^{w-1}
1-\rho(1-x_{z+j-w}^{(t-1)})) \right)\\
\leq \frac{1}{w}\e_{z}\lambda(\frac{1}{w}\sum_{j=0}^{w-1}
1-\rho(1-x_{z+j}^{(t-1)}))\leq
x_{\rm BP}/w,
\end{multline}
where the last inequality is obtained because $x_{z}^{(t-1)}\leq x_{\rm BP}$ for all $z$.
Now consider $\mathbf{y}$ to be $y_z = x_z^{(t)} + h (x_z^{(t+1)}-x_z^{(t)})$ for $0 \le h\le 1$. Thus,
\begin{equation*}
0 \leq x_z^{(t)}-y_z\leq h x_{\rm BP}/w. 
\end{equation*}

We write the second order Taylor expansion of $U(\mathbf{y};\e)$ around $U(\vx^{(t)};\e)$:
\begin{equation*}
U(\mathbf{y};\e)-U(\vx^{(t)}) = P(\mathbf{y},\vx^{(t)}) + R_2(\mathbf{y},\vx^{(t)}),
\end{equation*}
where $P(\cdot), R_2(\cdot)$ are defined in what
follows. We have
\begin{multline*}
P(\mathbf{y},\vx^{(t)}) = \sum_{z=1}^{N'}\frac{\partial U(\vx^{(t)};e)}{\partial x_z}  (y_z - x_z^{(t)})\\
+\frac12 \sum_{z=1}^{N'}\sum_{i=1}^{N'}\frac{\partial^2 U(\vx^{(t)};e)}{\partial x_z \partial x_i}(y_z-x_z^{(t)})(y_i-x_i^{(t)}),
\end{multline*}
leading to
\begin{multline*}
P(\mathbf{y},\vx^{(t)}) =\sum_{z=1}^{N'}\frac{\partial U(\vx^{(t)};e)}{\partial x_z}  (y_z - x_z^{(t)})\\
+\frac{1}{2}\sum_{z=1}^{N'} \rho'(1-x_z^{(t)})
(y_z-x_z^{(t)})^2 - \frac{1}{2} \sum_{z=1}^{N'} Q_z\\
- \frac{1}{2} \sum_{z=1}^{N'} \rho''(1-x_z^{(t)})(x_z^{(t)}-x_z^{(t+1)})
(y_z-x_z^{(t)})^2
\end{multline*}
and 
\begin{multline*}
\!\!\!Q_z \!=\! \frac{\rho'(1\!-\!x_z^{(t)})(y_z\!-\!x_z^{(t)})}{w}\sum_{k=0}^{w-1}
\e\lambda'(1\!-\!\frac{1}{w}\sum_{j=0}^{w-1} \rho(1\!-\!x_{z+j-k}^{(t)}))\times\\
\times\sum_{j=0}^{w-1} \frac{\rho'(1-x_{z+j-k}^{(t)})}{w}(y_{z+j-k}-x_{z+j-k}^{(t)}).
\end{multline*}
$R_2(\mathbf{y},\vx^{(t)})$ is the remainder term. According 
to Taylor's theorem, there exists $\mathbf{c}= \vx^{(t)}+ \gamma (\mathbf{y}-\vx^{(t)})$ for some $0\leq \gamma\leq1$ that
\begin{equation*}
R_2(\mathbf{y},\vx^{(t)}) = \frac{1}{6} \sum_{z,i,d} 
\frac{\partial^3 U(\mathbf{c};\e)}{\partial x_z \partial x_i\partial x_d}  
(y_z-x_z^{(t)})(y_i-x_i^{(t)})(y_d-x_d^{(t)})
\end{equation*}
where $z,i,d\in\{1,\dots,N'\}.$ The partial derivatives  
$\frac{\partial^3}{\partial x_z \partial x_i\partial x_d} U(\mathbf{c};\e)$, denoted by $U_{zid}(\mathbf{c};\e)$, are given in equations \eqref{eq:uzzz} to
\eqref{eq:uizd}. By rearranging the terms, we can write $R_2(\mathbf{y},\vx^{(t)})$ in the form of \eqref{eq:R2} in which there are two positive terms and 
two negative terms. Now, we show for a large enough $w$,
\begin{equation}
U(\mathbf{y};\e)-U(\vx^{(t)}) \leq 
\frac{1}{2} 
\sum_{z=1}^{N'}\frac{\partial U(\vx^{(t)};e)}{\partial x_z}  (y_z - x_z^{(t)}).
\label{eq:halfderiv}
\end{equation}
\begin{figure*}[t]
\begin{multline}
\label{eq:uzzz}
U_{zzz}(\mathbf{c};\e) = \rho'''(1-c_z) \left(c_z -\frac{1}{w}\sum_{k=0}^{w-1}\e_{z-k}\lambda\left(\frac{1}{w}\sum_{j=0}^{w-1} 1-\rho(1-c_{z+j-k})\right) \right) -2\rho''(1-c_z) +\\ 
3 \frac{\rho''(1\!-\!c_z)\rho'(1\!-\!c_z)}{w^2} \sum_{k=0}^{w-1} \e_{z-k}\lambda'\!\left(\frac{1}{w}\sum_{j=0}^{w-1} 1\!-\!\rho(1\!-\!c_{z+j-k})\right) - 
\frac{(\rho'(1\!-\!c_z))^3}{w^3} \sum_{k=0}^{w-1} \e_{z-k}\lambda'' \!\left(\frac{1}{w}\sum_{j=0}^{w-1} 1\!-\!\rho(1\!-\!c_{z+j-k})\right)\!,
\end{multline}
\begin{multline}
U_{zzi}(\mathbf{c};\e) =
\frac{\rho''(1-c_z)\rho'(1-c_i)}{w^2} \sum_{k=\max(0,z-i)}^{w-1+ \min(0,z-i)} \e_{z-k}\lambda'\left(\frac{1}{w}\sum_{j=0}^{w-1} 1-\rho(1-c_{z+j-k})\right)\\
 - \frac{(\rho'(1-c_z))^2 \rho'(1-c_i)}{w^3} \sum_{k=\max(0,z-i)}^{w-1+ \min(0,z-i)} \e_{z-k}\lambda''\left(\frac{1}{w}\sum_{j=0}^{w-1} 1-\rho(1-c_{z+j-k})\right),
\label{eq:uzzi}
\end{multline}
\begin{equation}
U_{zzi}(\mathbf{c};\e) =U_{ziz}(\mathbf{c};\e) =U_{izz}(\mathbf{c};\e),
\label{eq:uziz}
\end{equation}
\begin{equation}
U_{zid}(\mathbf{c};\e) =- \frac{\rho'(1-c_z) \rho'(1-c_i)\rho'(1-c_d)}{w^3} \sum_{k=\max(0,z-i,z-d)}^{w-1+ \min(0,z-i,z-d)} \e_{z-k}\lambda''\left(1-\frac{1}{w}\sum_{j=0}^{w-1} \rho(1-c_{z+j-k})\right),
\label{eq:uzid}
\end{equation}
\begin{equation}
U_{zid}(\mathbf{c};\e) =U_{zdi}(\mathbf{c};\e)=U_{dzi}(\mathbf{c};\e)=
U_{diz}(\mathbf{c};\e) =U_{izd}(\mathbf{c};\e) =U_{idz}(\mathbf{c};\e)
\label{eq:uizd}
\end{equation}
\begin{multline}
R_2(\mathbf{y},\vx^{(t)})= 
\frac{1}{6} \sum_{z=1}^{N'} \rho'''(1-c_z) \left(\!c_z \!-\!\frac{1}{w}\sum_{k=0}^{w-1}\e_{z-k}\lambda\left(1\!-\!\frac{1}{w}\sum_{j=0}^{w-1} \rho(1-c_{z+j-k})\right)\!\!\right)(y_z-x_z^{(t)})^3 - \frac{1}{3} \sum_{z=1}^{N'} \rho''(1-c_z)(y_z-x_z^{(t)})^3-\\
-
\frac{1}{6} \sum_{z=1}^{N'} 
\frac{\rho'(1-c_z)(y_z-x_z^{(t)})}{ w}\sum_{k=0}^{w-1}
\e_{z-k}\lambda''\left(1-\frac{1}{w}\sum_{j=0}^{w-1} \rho(1-c_{z+j-k})\right)\times
\left(\sum_{j=0}^{w-1} \frac{\rho'(1-c_{z+j-k})}{w}(y_{z+j-k}-x_{z+j-k}^{(t)})\right)^2\\
+\frac{3}{6} 
\label{eq:R2}\sum_{z=1}^{N'}
\frac{\rho''(1-c_z)(y_z-x_z^{(t)})^2}{ w}\sum_{k=0}^{w-1}
\e_{z-k}\lambda'\left(1-\frac{1}{w}\sum_{j=0}^{w-1} \rho(1-c_{z+j-k})\right)\times
\left(\sum_{j=0}^{w-1} \frac{\rho'(1-c_{z+j-k})}{w}(y_{z+j-k}-x_{z+j-k}^{(t)})\right).
\end{multline}
\hrulefill
\end{figure*}
We know that $\frac{\partial}{\partial x_z} U(\vx^{(t)};e) = \rho'({1-x^{(t)}})(x^{(t)}-x^{(t+1)})$ and then,
\begin{multline*}
P(\mathbf{y},\vx^{(t)})\leq\frac{1}{2} 
\sum_{z=1}^{N'}\frac{\partial U(\vx^{(t)};e)}{\partial x_z}  (y_z - x_z^{(t)})- \frac{1}{2} \sum_{z=1}^{N'} Q_z\\
- \frac{1}{2} \sum_{z=1}^{N'} \rho''(1-x_z^{(t)})(x_z^{(t)}-x_z^{(t+1)})
(y_z-x_z^{(t)})^2
\end{multline*}
We have the result if we show that
\begin{multline}
R_2(\mathbf{y},\vx^{(t)})- \frac{1}{2} \sum_{z=1}^{N'} Q_z\\
- \frac{1}{2} \sum_{z=1}^{N'} \rho''(1-x_z^{(t)})(x_z^{(t)}-x_z^{(t+1)})
(y_z-x_z^{(t)})^2\leq 0.
\label{eq:negative}
\end{multline}
For each $z$, the second term and the third term of this expression are negative. Furthermore, first term and the last term of $R_2(\mathbf{y},\vx^{(t)})$ in \eqref{eq:R2} are also negative (note that $y_z-x_z^{(t)} \leq 0$).
Finally, we consider the positive terms of $R_2(\mathbf{y},\vx^{(t)})$, which are, for each $z$,
\begin{equation}
-\frac{1}{3} \rho''(1-c_z)(y_z-x_z^{(t)})^3,
\label{eq:pos1}
\end{equation}
and
\begin{multline}
\!\!-\frac{\rho'(1\!-\!c_z)(y_z\!-\!x_z^{(t)})}{6 w}\sum_{k=0}^{w-1}
\e\lambda''\left(\!1\!-\!\frac{1}{w}\sum_{j=0}^{w-1} \rho(1\!-\!c_{z+j-k})\right)\times
\\
\times\left(\sum_{j=0}^{w-1} \frac{\rho'(1-c_{z+j-k})}{w}(y_{z+j-k}-x_{z+j-k}^{(t)})\right)^2.
\label{eq:pos2}
\end{multline}
Now compare \eqref{eq:pos1} and
\begin{equation}
\frac{1}{2}\rho''(1-x_z^{(t)})(x_z^{(t)}-x_z^{(t+1)})
(y_z-x_z^{(t)})^2.
\end{equation}
For a large $w$ the above term dominates \eqref{eq:pos1} since
\begin{equation}
x_z^{(t)} - h x_{\rm BP}/w\leq y_z\leq c_z\leq x_z^{(t)}.
\end{equation}
Now consider the pair of $Q_z$ and \eqref{eq:pos2} for each $z$. 
We have the bound 
\begin{multline*}
\left\vert\sum_{j=0}^{w-1} \frac{\rho'(1-c_{z+j-k})}{w}(y_{z+j-k}-x_{z+j-k}^{(t)})\right\vert\\
< h\frac{x_{\rm BP}}{w}\sum_{j=0}^{w-1} \frac{\rho'(1-c_{z+j-k})}{w}
< h\frac{x_{\rm BP}}{w}\rho'(1).
\end{multline*}
$Q_z$ is  positive and $O(\frac{1}{w^2})$ for a large $w$ but
\eqref{eq:pos2} is $O(\frac{1}{w^3})$ and thus, there exists $w'$
such that for $w>w'$, $Q_z$ becomes larger than \eqref{eq:pos2}.

Consolidating the above observations provides the desired inequality
\eqref{eq:negative} for a large $w$ and hence, 
\eqref{eq:halfderiv} follows.

\section{Proof of Theorem \ref{thm:upperbound1}}\label{apn:upperbound1}
Recall that $x_z^{(t)}$ is an increasing sequence in terms of $z$ and a
decreasing sequence in terms of $t$.
We assume that the DE solution moves uniformly, i.e. there is $T\in \mathbb{N}$ such that
\begin{equation*}
x_{z}^{(t+T)}\leq x_{z-1}^{(t)}<x_z^{(t+T-1)},
\end{equation*}
for all $t\in\mathbb{N}$ and $z\in\mathbb{Z}$. Without loss of generality, assume that $t=0$ and define $y_z = x_{z}^{(0)}$. Thus,
\begin{equation}
y_z\!=\!x_z^{(0)}\!\geq\! x_z^{(1)}\!\geq\cdots\geq\! x_z^{(T-1)}\!>\! x_{z-1}^{(0)}\!=\!y_{z-1}
\!\geq\! x_z^{(T)}.
\label{eq:update_order}
\end{equation}
Let $\vx^{(t)}$ denote the sequence $x_z^{(t)}$ at iteration $t$ and let
 $\vy^{[k]}$ denote the $k$-shifted sequence $y_{z}^{[k]}=y_{z-k}$ for $k\in\mathbb{Z}$. 
According to Lemma \ref{lem:taylor},
\begin{equation*}
\alpha(U(\vx^{(t+1)};\e)-U(\vx^{(t)};\e)) \leq  -\sum_{z=1}^{N'} \rho'(1-x_z^{(t)})
(x_z^{(t)}-x_z^{(t+1)})^2.
\end{equation*}
Thus, we can write
\begin{align*}
U(\vy^{[1]};\e)-U(\vy^{[0]};\e)
&=\sum_{t=0}^{T-2} U(\vx^{(t+1)};\e)-U(\vx^{(t)};\e)\\
&+ U(\vy^{[1]};\e) - U(\vx^{(T-1)};\e),
\end{align*}
and hence 
\begin{align*}
\alpha(U(\vy^{[1]}&;\e)\!-\!U(\vy^{[0]};\e))
\leq -\!\sum_{t=0}^{T-2} \sum_{z=1}^{N'} \rho'(1\!-\!x_z^{(t)})
(x_z^{(t)}\!-\!x_z^{(t+1)})^2\\
&-\sum_{z=1}^{N'}\rho'(1\!-\!x_z^{(T-1)})(x_z^{(T-1)}\!-\!y_{z-1})(x_z^{(T-1)}\!-\!x_z^{(T)}),
\end{align*}
where in the last application of Lemma~\ref{lem:taylor}, we have used $h = ( x_z^{(T-1)}-y_{z-1})/(x_z^{(T-1)}-x_z^{(T)})$ since we assume that
the DE solution moves uniformly.
 Since $y_{z-1}\geq x_z^{(T)}$, we have
 \begin{multline*}
\alpha(U(\vy^{[1]};\e)\!-\!U(\vy^{[0]};\e))
\leq -\! \sum_{t=0}^{T-2} \sum_{z=1}^{N'} \rho'(1\!-\!x_z^{(t)})
(x_z^{(t)}\!-\!x_z^{(t+1)})^2\\
-\!\sum_{z=1}^{N'}\rho'(1\!-\!x_z^{(T-1)})(x_z^{(T-1)}\!-\!y_{z-1})^2,  
\end{multline*}
or equivalently,
 \begin{multline*}
\alpha(U(\vy^{[0]};\e)-U(\vy^{[1]};\e))
 \geq  \sum_{t=0}^{T-2} \sum_{z=1}^{N'}  \rho'(1-x_z^{(t)})
(x_z^{(t)}-x_z^{(t+1)})^2\\
+\sum_{z=1}^{N'}\rho'(1-x_z^{(T-1)})(x_z^{(T-1)}-y_{z-1})^2,  
\end{multline*}
Now consider the right hand side of the above inequality and let us
exchange the order of the sums. We can change it even for
$N'\to\infty$ because the summands are positive. Then for each $z$, we
minimize the following summation:
\begin{align*}
 \sum_{t=0}^{T-2}  &\rho'(1-x_z^{(t)})
(x_z^{(t)}-x_z^{(t+1)})^2\\
&+\rho'(1-x_z^{(T-1)})(x_z^{(T-1)}-y_{z-1})^2
=   \sum_{t=0}^{T-1} \rho'(1-x_z^{(t)}) l_t^2
\end{align*}
where $l_t = x_z^{(t)}-x_z^{(t+1)}$
 for $0 \leq t < T-1$ and $l_{T-1} = x_z^{(T-1)}-y_{z-1}$.
 Since $\rho'(1-x)$ is a decreasing function, we have
\begin{equation*}
\sum_{t=0}^{T-1} \rho'(1-x_z^{(t)}) l_t^2\geq 
\rho'(1-y_{z}) \sum_{t=0}^{T-1} l_t^2.
\end{equation*}
Due to \eqref{eq:update_order}, $l_t\geq 0$ and
\begin{equation*}
\sum_{t=0}^{T-1} l_t = x_{z}^{(0)}-y_{z-1} = y_{z}-y_{z-1}.
\end{equation*}
The application of Jensen's inequality leads to
\begin{equation*}
\frac{1}{T} \sum_{t=0}^{T-1} l_t^2\geq \left(\frac{1}{T}\sum_{t=0}^{T-1} l_t \right)^2 = \frac{1}{T^2}(y_{z}-y_{z-1})^2 ,
\end{equation*}
and thus,
\begin{equation}
\alpha(U(\vy^{[0]};\e)-U(\vy^{[1]};\e)) \geq \frac{1}{T} \sum_{z=1}^{N'}     \rho'(1-y_{z})(y_{z}-y_{z-1})^2.
\label{eq:1}
\end{equation}
Considering the definition of the potential function in \eqref{eq:potcoupl}, we observe that $U(\vy^{[0]};\e)-U(\vy^{[1]};\e)$ 
is a telescoping sum with only the difference between the first and the last term remaining, i.e.
\begin{multline*}
U(\vy^{[0]};\e)-U(\vy^{[1]};\e) = 
\frac{1}{R'(1)} \left(1-R(1-y_{N'})\right)- \\
- y_{N'}\rho(1-y_{N'})
-\frac{1}{R'(1)} \left(1-R(1-y_{0})\right) + y_{0}\rho(1-y_{0})-\\
-\frac{\e}{L'(1)} L\left(1-\frac{1}{w}\sum_{j=0}^{w-1}
\rho(1-y_{N'+j})\right)-\\
+\frac{\e}{L'(1)} L\left(1-\frac{1}{w}\sum_{j=0}^{w-1}
\rho(1-y_{j})\right).
\end{multline*}
With $y_{z}=0$ for $z\leq w$ (by the theorem assumption) and $y_{z}=y_{N'}=x_{\rm BP}$ for $z>N'$, we have
\begin{equation}
U(\vy^{[0]};\e)-U(\vy^{[1]};\e) = U(x_{\rm BP};\e)\,.
\label{eq:2}
\end{equation}
Finally, the statement of the theorem is concluded from \eqref{eq:1} and \eqref{eq:2}.

\section{Proof of Theorem \ref{thm:upperbound2}}\label{apn:upperbound2}

In order to prove the theorem, we require the following useful lemma.
\begin{lem}\label{lem:diff_bound}
Let $x_z^{(t)}$ be the solution of DE equation \eqref{eq:decoupl} at some iteration $t$. Further assume that $T_1 > 1$. Then,
\begin{equation}
\label{eq:diff_bound}
x_z^{(t)}-x_{z-1}^{(t)} \geq \left\vert\frac{x_z^{(t)}-\e\lambda\left(1-\rho(1-x_z^{(t)})\right)}{w}\right\vert,
\end{equation}
for $z\leq N'$.
\end{lem}
\begin{proof}First we prove that $x_z^{(t)}-x_{z-1}^{(t)} \geq \left(x_z^{(t)}-\e\lambda\left(1-\rho(1- x_z^{(t)})\right)\right)/w$.
Define
\begin{equation*}
\nu_z^{(t)}:=1 - \frac{1}{w}\sum_{j=0}^{w-1} \rho(1-x_{z+j}^{(t)}).
\end{equation*}
Note that $\nu_z^{(t)}$ is an increasing sequence in terms of $z$, because $1-\rho(1-x)$ is an increasing function and $x_z^{(t)}$ is an increasing sequence in $z$.
Thus, according to \eqref{eq:decoupl},
\begin{align*}
x_{z}^{(t)} = \frac{1}{w}\sum_{k=0}^{w-1} \e_{z-k} \lambda(\nu_{z-k}^{(t-1)})\leq \e_{z} \lambda(\nu_{z}^{(t-1)}),
\end{align*}
and by rewriting the above equation,
\begin{align*}
(w-1) x_{z}^{(t)}&\geq \sum_{k=0}^{w-1} \e_{z-k} \lambda(\nu_{z-k}^{(t-1)}) -\e_{z} \lambda(\nu_{z}^{(t-1)})\\
&=\sum_{k=1}^{w-1} \e_{z-k} \lambda(\nu_{z-k}^{(t-1)})=
\sum_{k=0}^{w-2} \e_{z-1-k} \lambda(\nu_{z-1-k}^{(t-1)})
\end{align*}
Thus, we can upper bound $wx_{z-1}^{(t)}$ as
\begin{align}
w x_{z-1}^{(t)} 
&= \sum_{k=0}^{w-1} \e_{z-1-k} \lambda(\nu_{z-1-k}^{(t-1)})\nonumber\\
&=  \sum_{k=0}^{w-2} \e_{z-1-k} \lambda(\nu_{z-1-k}^{(t-1)})+
\e_{z-w} \lambda(\nu_{z-w}^{(t-1)})\nonumber\\
&\leq (w-1)x_{z}^{(t)} + \e_{z-w} \lambda(\nu_{z-w}^{(t-1)})\nonumber\\
&\leq  (w-1)x_{z}^{(t)} + \e \lambda\left(1-\rho(1-x_{z-1}^{(t-1)})\right)\nonumber\\
&\leq  (w-1)x_{z}^{(t)} + \e \lambda\left(1-\rho(1-x_{z}^{(t)})\right).
\label{eq:3}
\end{align}
The last two inequalities result from the facts that $\e_{z-w}\leq \e$ and that
\begin{equation*}
\nu_{z-w}^{(t-1)}=1-\frac{1}{w}\sum_{j=0}^{w-1} \rho(1-x_{z-w+j}^{(t-1)})
\leq 1-\rho(1-x_{z-1}^{(t-1)}).
\end{equation*}
Furthermore, $x_{z-1}^{(t-1)}<x_{z}^{(t)}$ since $T_1>1$.
The result is given by rearrangement of \eqref{eq:3}. Similarly, we prove that
$x_z^{(t)}-x_{z-1}^{(t)} \geq \left(\e\lambda\left(1-\rho(1- x_z^{(t)})\right)-x_z^{(t)}\right)/w$.
\begin{align*}
x_{z-1}^{(t)} = 
\frac{1}{w} \sum_{k=0}^{w-1} \e_{z-1-k} \lambda(\nu_{z-1-k}^{(t-1)})
\geq \e_{z-w} \lambda(\nu_{z-w}^{(t-1)}),
\end{align*}
and thus,
\begin{align*}
(w-1)x_{z-1}^{(t)} &\leq  
\sum_{k=0}^{w-1} \e_{z-1-k} \lambda(\nu_{z-1-k}^{(t-1)})
- \e_{z-w} \lambda(\nu_{z-w}^{(t-1)})\\
&= \sum_{k=0}^{w-2} \e_{z-1-k} \lambda(\nu_{z-1-k}^{(t-1)})
=\sum_{k=1}^{w-1} \e_{z-k} \lambda(\nu_{z-k}^{(t-1)}).
\end{align*}
For $x_z^{(t)}$, we have
\begin{align*}
wx_z^{(t)} &= \sum_{k=1}^{w-1} \e_{z-k} \lambda(\nu_{z-k}^{(t-1)}) +
\e_{z} \lambda(\nu_{z}^{(t-1)})\\
&\geq (w-1)x_{z-1}^{(t)} +\e_{z} \lambda(\nu_{z}^{(t-1)})\\
&\geq (w-1)x_{z-1}^{(t)} +\e \lambda(1-\rho(1-x_z^{(t-1)}))\\
&\geq (w-1)x_{z-1}^{(t)} +\e \lambda(1-\rho(1-x_z^{(t)})).
\end{align*}
We conclude the result since
\begin{align*}
(w+1)x_z^{(t)}\geq w x_{z-1}^{(t)} +\e \lambda(1-\rho(1-x_z^{(t)})).
\end{align*}

\end{proof}
We now turn to the proof of Theorem \ref{thm:upperbound2}. To derive an
upper bound on the convergence speed, we must lower-bound the following summation:
\begin{equation}
\sum_{z=1}^{N'} \rho'(1-x_z^{(t)})(x_z^{(t)}-x_{z-1}^{(t)})^2.
\label{eq:profile_sum}
\end{equation}
Let $y_z = x_z^{(t)}$ for $z\in\mathbb{Z}$.
From Lemma \ref{lem:pot_properties} and Lemma \ref{lem:diff_bound}, we have
\begin{align*}
\rho'(1-y_z)(y_z-y_{z-1})\geq  \left\vert\frac{\partial U(y_z;\e)}{w\partial y_z}\right\vert =: \left\vert\frac{U'(y_z;\e)}{w}\right\vert.
\end{align*}
We can divide the interval $[0,x_s]$
into four sub-intervals (see also Fig.~\ref{fig:pot}):
\begin{itemize}
\item[i)] $x\in[0,x_1]$ where $U'(x;\e)\geq 0$ and $U''(x;\e)\geq 0$.
\item[ii)] $x\in(x_1,x_u]$ where $U'(x;\e)\geq 0$ and $U''(x;\e)\leq 0$.
\item[iii)] $x\in(x_u,x_2]$ where $U'(x;\e)\leq 0$ and $U''(x;\e)\leq 0$.
\item[iv)] $x\in(x_2,x_s]$ where $U'(x;\e)\leq 0$ and $U''(x;\e)\geq 0$.
\end{itemize}
Let $a$ denote the index where $y_a\leq x_1<y_{a+1}$. Similarly, we define  $b$ and $c$
for $x_u$ and $x_2$, respectively. We split \eqref{eq:profile_sum} into four sum according to the sun intervals and lower-bound them separately.

\textbf{i)} For $z\leq a$, since $U''(y_z;\e)\geq 0$ ($\cup$-convexity), $U'(y_z;\e)(y_z-y_{z-1})\geq U(y_z;\e)-U(y_{z-1};\e)$. Thus,
\begin{equation*}
\sum_{z=1}^{a} \vert U'(y_z;\e)\vert (y_z-y_{z-1})\geq U(y_a;\e).
\end{equation*}

\textbf{ii)} Since the potential functions becomes $\cap$-convex for ${a+1}\leq z\leq b$,
\begin{equation*} 
U(y_z;\e)-U(y_{z-1};\e)\leq
U'(y_{z-1};\e)(y_z-y_{z-1}).
\end{equation*}
According to mean value theorem, $\exists k_z\in[y_{z-1},y_z]$ such that
$$U'(y_z;\e) = U'(y_{z-1};\e) + U''(k_z;\e) (y_z-y_{z-1}).$$

Let $D=\max_{x\in[0,x_s]}\vert U''(x;\e)\vert$. Then,
\begin{align*}
U'(y_z;\e)(y_z\!-\!y_{z-1}) &\!\geq\! U'(y_{z-1};\e)(y_z\!-\!y_{z-1}) \!-\! D (y_z\!-\!y_{z-1})^2\\
&\begin{multlined}[b][0.65\columnwidth]
\!\geq\! U(y_z;\e)-U(y_{z-1};\e) \\-D \frac{x_s}{w}(y_z-y_{z-1}).
\end{multlined}
\end{align*}
The last inequality is due to $y_z-y_{z-1}\leq x_s/w$ which can be
shown by writing the DE equation for $y_z$ and $y_{z-1}$ and then,
since $y_z-y_{z-1}$ is a telescoping sum, the inequality follows. Hence,
\begin{align*}
\sum_{z=a+1}^{b} \vert U'(y_z;\e)\vert (y_z-y_{z-1})\geq 
U(y_b;\e)&-U(y_{a};\e) \\
&-D \frac{x_s}{w}(y_b-y_{a})
\end{align*}

\textbf{iii)} Similar to i) we now use the $\cap$-convexity property. For $b+1<z\leq c$, $U'(y_z;\e)(y_z-y_{z-1})\leq U(y_z;\e)-U(y_{z-1};\e)$
since $U''(y_z;\e)\leq 0$ and then,
\begin{align*}
\sum_{z=b+1}^{c} \vert U'(y_z;\e)\vert (y_z-y_{z-1})
&=-\sum_{z=b+1}^{c} U'(y_z;\e) (y_z-y_{z-1})\\
&\geq U(y_{b};\e)- U(y_c;\e).
\end{align*}

\textbf{iv)} For $z>c$, $U''(y_z;\e)\geq 0$, which implies 
$$U'(y_{z-1};\e)(y_z-y_{z-1})\leq U(y_z;\e)-U(y_{z-1};\e).$$
By the mean value theorem, we get
\begin{equation*}
U'(y_z;\e) \leq U'(y_{z-1};\e) + D (y_z-y_{z-1}),
\end{equation*}
and consequently,
\begin{equation*}
U'(y_z;\e) (y_z-y_{z-1})\leq U(y_z;\e)-U(y_{z-1};\e) + D\frac{x_s}{w} (y_z-y_{z-1}),
\end{equation*}
and,
\begin{align*}
\sum_{z=c+1}^{N'} \vert U'(y_z;\e)\vert (y_z-y_{z-1})
&=-\sum_{z=c+1}^{N'} U'(y_z;\e) (y_z-y_{z-1})\\
&\begin{multlined}[b][0.55\columnwidth]
\geq U(y_{c};\e)- U(y_{N'};\e)\\
-D\frac{x_s}{w}(y_{N'}-y_c),
\end{multlined}
\end{align*}
where $y_{N'}=x_{s}$. By summing up all terms,
\begin{multline*}
\sum_{z=1}^{N'} \vert U'(y_z;\e)\vert (y_z-y_{z-1})\geq
2U(y_b;\e)-U(x_s;\e)\\
 - D\frac{x_s}{w}(x_s-y_c + y_b-y_a).
\end{multline*}
by definition $y_b\leq x_u<y_{b+1}$. Since $U'(x_u;\e)=0$, there exists
$k\in[y_b,x_u]$ such that,
\begin{align*}
U(y_b;\e) &= U(x_u;\e) +\frac{U''(k;\e)}{2}(x_u-y_b)^2\\
&\geq U(x_u;\e) -\frac{D}{2}(x_u-y_b)^2\\
&\geq U(x_u;\e) -\frac{D}{2}\frac{x_s}{w}(x_u-y_b).
\end{align*}
The last inequality holds since $x_u-y_b\leq y_{b+1}-y_{b}\leq x_s/w$.
Finally,
\begin{align*}
\sum_{z=1}^{N'} \vert U'(y_z;\e)\vert (y_z-y_{z-1})\geq&
2U(x_u;\e)-U(x_s;\e)\\
& - D\frac{x_s}{w}(x_s-y_c + x_u-y_a)\\
\geq& 2U(x_u;\e)-U(x_s;\e) - D\frac{x_s^2}{w}
 .
\end{align*}
We have the result by noting that
\begin{equation*}
\sum_{z=1}^{N'} \rho'(1-x_z^{(t)})(x_z^{(t)}-x_{z-1}^{(t)})^2
\geq \frac{1}{w} 
\sum_{z=1}^{N'} \vert U'(y_z;\e)\vert (y_z-y_{z-1})\,.
\end{equation*}

\end{appendices}

\section*{Acknowledgment}
We thank R\"udiger Urbanke and Nicolas Macris for insightful discussions. Vahid Aref was
supported by grant No. 200021-125347 of the Swiss National Science
Foundation.
Laurent Schmalen was supported by the German Government in the frame of the CELTIC+/BMBF project SASER-SaveNet.

\footnotesize
\bibliographystyle{IEEEtran}
\bibliography{references}
\end{document}